\documentclass{amsproc}

\newtheorem{theorem}{Theorem}[section]

\newtheorem{proposition}[theorem]{Proposition}

\theoremstyle{definition}

\newtheorem{example}[theorem]{Example}

\theoremstyle{remark}

\newtheorem{remarks}[theorem]{Remarks}

\numberwithin{equation}{section}

\begin{document}

\title{Momentum operators on graphs}

\author{Pavel Exner}
\address{Department of Theoretical Physics, Nuclear Physics
Institute, Czech Academy of Sciences, 25068 \v{R}e\v{z} near Prague, and
Doppler Institute for Mathematical Physics and Applied
Mathematics, Czech Technical University, B\v{r}ehov\'{a}~7, 11519
Prague, Czechia}
\email{exner@ujf.cas.cz}
\thanks{The research was supported by the Czech Science Foundation within the project P203/11/0701. I am indebted to Peter Kuchment for a discussion on his book in preparation and to the referee for reading the manuscript carefully.}

%    General info
\date{January 1, 1994 and, in revised form, June 22, 1994.}

\dedicatory{With congratulations to my friend Fritz Gesztesy on the occasion of his 60th birthday.}


\begin{abstract}
We discuss ways in which momentum operators can be introduced on an oriented metric graph. A necessary condition appears to the balanced property, or a matching between the numbers of incoming and outgoing edges; we show that a graph without an orientation, locally finite and at most countably infinite, can made balanced oriented \emph{iff} the degree of each vertex is even. On such graphs we construct families of momentum operators; we analyze their spectra and associated unitary groups. We also show that the unique continuation principle does not hold here.
\end{abstract}

\maketitle

\section{Introduction}

Our writings have their own fates once they left our hands and one can only guess how successful they will eventually be. Fritz bibliography is extensive, some two hundred items if not more, and covers many areas. All his texts are pleasure to read, deep in contents and perfectly organized. Nevertheless, one of them made much larger impact than any others, namely the monograph \cite{AGHH} first published in 1988. It is a collective work but Fritz hand in unmistakably present in the exposition, and I add that it makes me proud to be a part of the second edition.

Thinking about how the book resonated in the community, one has to come to the conclusion that that it did not happen by a chance, rather the subject of solvability struck some important needs. Such models are complex and versatile enough to be applicable to numerous physical situations and at the same time, they are mathematically accessible and allow derive conclusions without involving a heavy machinery. This concerns both the proper subject of the book \cite{AGHH} and its various extensions, among them the theory of quantum graphs which is developing rapidly; a broad overview of recent developments can be found in \cite{EKKST}.

This is the frame into which the present little \emph{\'etude} belongs. While from the quantum mechanical point of view the operators most frequently studied are Hamiltonians which are, mathematically speaking, typically Laplacians, their singular perturbations and modifications, there are other observables of interest too. In this paper we deal with momentum operators acting locally as imaginary multiples of the first derivative in the appropriate functional spaces.

This is not to say that first-order operators on graphs have not been studied before. On the one hand, momentum-type operators on graphs have been discussed recently in \cite{FKW} from the viewpoint of appropriate index theorems, see also \cite{BK} where their use in quantum-chaos studies based on bond scattering matrices is described. Momentum operators are considered also in \cite{Ca, Eg}, see \cite{ES10} for related analysis of Berry-Keating operator. On the other hand, there is a recent series of voluminous papers treating such operators on two or a larger number of intervals \cite{JPT11, JPT12a, JPT12b}; they come from a different literature background, only weakly connected to the applications of self-adjoint extensions in modeling quantum dynamics such as those mentioned above or collected in \cite{AEG}.

Our discussion in this paper necessarily overlaps in part with the indicated studies looking at the problem from a bit different perspective. In contrast to graph Laplacians definition of momentum operators require the graph to be oriented, not just arbitrarily but in a \emph{balanced} way, with matching numbers of incoming and outgoing edges in each vertex. We show that an undirected graph, locally finite and at most countably infinite, can be given such orientations, in general different ones, if and only if the degrees of its vertices are even, and describe possible momentum operators. We derive some properties of their spectra, describe unitary groups associated with them, and present an example showing that the unique continuation principle is not valid here.

\section{Preliminaries}

Let us first introduce some notions we will need. We consider a graph $\Gamma$ consisting of a family of \emph{vertices}, $\mathcal{V}=\{v_j:\: j\in I^v\}$ indexed by a set $I\equiv I^v$, and a family $\mathcal{E}$ of edges\footnote{The graph constituent symbols can be labelled by $\Gamma$ but we mostly refrain from doing that.}. The latter includes \emph{finite} (or \emph{internal}) edges, $\mathcal{E}_\mathrm{fin}= \{e_k:\: k\in I^e_\mathrm{fin}\}$ and \emph{semi-infinite} edges (alternatively \emph{external} ones or \emph{leads}), $\mathcal{E}_\infty = \{h_k:\: k\in I^e_\infty\}$. If each pair of vertices is connected by at most one edge, we can identify $I^e_\mathrm{fin}$ with a subset of $I\times I$ specified the adjacency matrix, and if no more than one lead is attached to each vertex, $I^e_\infty$ can be identified with a subset of $I^v$. Note that a graph can be  always modified to satisfy these requirements by inserting dummy vertices to the ``superfluous'' edges, however, we will not need these assumptions.

The graphs we are going to consider are \emph{metric} and \emph{oriented}. The first notion means that each finite edge $e_k$ can be identified with a line segment $[0,l_k]$ and a semi-infinite one with a halfline. The second one says that the orientation of the edge parametrization is not arbitrary; each finite edge $e_k$ has its \emph{starting point} associated with $x=0$ and \emph{endpoint} to which $x=l_k$ corresponds. In the standard terminology of graph theory oriented graph form a subclass among the \emph{directed} ones in which each edge has a single orientation. In contrast to the usual graph theory, our graphs can also have semi-infinite edges and we have to take account of them. We divide $\mathcal{E}_\infty$ into the family $\mathcal{E}_\infty^+ = \{h_k:\, k\in I^e_{\infty,+} \} $ of \emph{outgoing} edges parametrized by $[0,\infty)$, with zero referring to the vertex, and $\mathcal{E}_\infty^-$ consisting of \emph{incoming} ones, parametrized by $(-\infty,0]$.

We can thus count edges meeting at a given vertex $v_j$: there are $d_j^{\mathrm{fin},\pm}$ finite edges starting and ending there, respectively, together $d_j^{\infty,+}$ outgoing and $d_j^{\infty,-}$ incoming leads. The vertex degree $d_j$ is the sum of those four numbers which can be split into the incoming and outgoing part, $d_j^\pm:= d_j^{\mathrm{fin},\pm} + d_j^{\infty,\pm}$. We also introduce $d_j^{\mathrm{fin}} := d_j^{\mathrm{fin},+} +d_j^{\mathrm{fin},-}$ and $d_j^{\infty} := d_j^{\infty,+} +d_j^{\infty,-}$ and set
 % ------------- %
\begin{equation}
N_\mathrm{fin}:= \sum_{j\in I} d_j^{\mathrm{fin},+} = \sum_{j\in I} d_j^{\mathrm{fin},-}\,, \quad N_\infty^\pm:= \sum_{j\in I} d_j^{\infty,\pm}\,, \quad N_\infty = N_\infty^+ + N_\infty^- \,.
\end{equation}
 % ------------- %
In all the paper we suppose that $\Gamma$ is either \emph{locally finite} and \emph{at most countably infinite} meaning that $I$ is at most countable and
 % ------------- %
\begin{equation} \label{finitedeg}
\exists\, c_d>0\::\;\; d_j\le c_d \quad \mathrm{for} \;\; \forall j\in I\,.
\end{equation}
 % ------------- %
This includes several categories, in particular, \emph{finite} graphs having $\sharp I<\infty$ and $N_\infty=0$, further \emph{finite-core} graphs with $\sharp I<\infty$ and $N_\infty>0$, and finally \emph{properly countably infinite} meaning that $I$ is countable. In the latter case the values of $N_\mathrm{fin}$ and $N_\infty$ can be zero, a finite number, or infinite, in all possible combinations. In case of infinite graphs we shall also suppose that.
 % ------------- %
\begin{equation} \label{minlength}
\inf_{j\in I} l_j > 0\,.
\end{equation}
 % ------------- %
The graph will be called \emph{balanced}, or alternatively \emph{balanced oriented}, provided
 % ------------- %
\begin{equation}
d_j^{\mathrm{fin},+} = d_j^{\mathrm{fin},-} \;\; \mathrm{and} \quad d_j^{\infty,+} = d_j^{\infty,-} \quad \mathrm{for} \;\; \forall j\in I\,;
\end{equation}
 % ------------- %
the said notion is again common in oriented graph theory and we remark that our definition is consistent with the standard terminology since one can amend our graphs having $N_\infty^\pm\ne 0$ with a vertex at infinity where the outgoing semi-infinite edges ``end'' and the incoming ones ``start''.

\section{Balanced orientability} \label{s: balanced}

In physical models where graphs are employed to describe spatial structures and motion of particles or fields on them usually no orientation is prescribed. Before proceeding further let us thus ask whether and how an undirected graph can be given orientation. The question has to be made more precise because one can always put arrows to such a graph edges; in graph theory one usually asks about existence of oriented paths connecting any pair of graph vertices. In contrast to that for us the local balance will matter; we call the a metric graph \emph{balanced orientable} if one can parametrize its edges in a way which makes it an oriented and balanced graph.
 % ------------- %
\begin{theorem}
An undirected graph $\Gamma$ satisfying the above countability assumption and (\ref{finitedeg}) is balanced orientable \emph{iff} the degree of any vertex $v\in\mathcal{V}_\Gamma$ is even.
\end{theorem}
 % ------------- %
\begin{proof}
The necessary condition is obvious. To check the sufficient one we introduce some notions. A \emph{path} in an oriented graph is a family of subsequently adjacent edges such that at each vertex on it an incoming edge meets an outgoing one; the edges constituting a path can be parametrized by adjacent intervals of the real axis. A \emph{free path} is a path which can be followed in both directions without termination; it can be either \emph{infinite} parametrizable by $\mathbb{R}$ or a \emph{loop} consisting of a finite number of finite edges which brings one to the initial point.

By assumption the vertices $v\in\mathcal{V}_\Gamma$ can be numbered. Take $v_1\in\mathcal{V}_\Gamma$, pick two edges emanating from it and give them orientations making one of them incoming and the other outgoing. Follow the latter to the other end, choose one of the (even number of the) edges emanating from that vertex and make it outgoing, and do the same ``backwards'' with the chosen edge incoming at $v_1$. Proceeding in the same way we get either an infinite free path or a loop in case the forward and backward branch meet at a vertex; it is clear from the construction that the loop will have a definite orientation. We delete the constructed free path from $\Gamma$ obtaining a graph $\Gamma'$ which satisfies the assumptions of the theorem having $d'_1\le d_1-2$ and $d'_j \le d_j$ for $j\ge 2$. If $d'_1\ne 0$ we construct in $\Gamma'$ another free path through $v_1$, delete it from the graph, and continue until this vertex is eliminated entirely.

Keeping the original vertex numbering we proceed to the vertex with lowest index value in the ``reduced'' graph obtained in this way, construct a free path going through it, eliminate it from the graph, and go on with such path eliminations until this vertex is fully removed from the graph; continuing the procedure we exhaust after an at most countable number of steps all the vertices which reduces the undirected graph remainder to an empty set proving thus the result.
\end{proof}

 % ------------- %
\begin{remarks} \label{rmk:orientability}
(a) If $N^\mathrm{fin}<\infty$ the graph balanced orientability requires the number of semi-infinite edges to be even. This is seen well if we use the ``flower'' model \cite{Ku, EL} of $\Gamma$ with all the vertices put together. Since the degree of this ``grand vertex'' is the sum of degrees of the original vertices and every finite edge contributes to this quantity by an even number, and furthermore, the number of external edges has to be finite in view of (\ref{finitedeg}), it must be even. Another simple consequence of the above result is that a balanced orientable graph cannot have ``loose ends'', i.e. finite edges with a vertex of degree one.

(b) The simplest situation from the orientability point view occurs if the graph edges can be paired; this is the case, in particular, if one thinks of an undirected graph edge as of a pair of \emph{bonds} having opposite orientations. This is a useful trick, employed for instance in studies of quantum chaos on graphs \cite{KS}, allowing one to express properties of the original undirected graph in terms of ``bond scattering matrices'' -- cf.~\cite[Sec.~2.2]{BK}.

(c) The construction employed in the proof shows that a balanced oriented graph can be regarded as \emph{a union of free paths}, disjoint except for the vertices in which they intersect. In general, an undirected graph satisfying the assumptions of the theorem can be oriented in different ways. In case of a \emph{free chain graph} with all the vertices of degree two the ambiguity is trivial consisting of a choice of one the two possible orientations, while in case of nontrivial \emph{branchings} meaning existence of a $j\in I$ with $d_j>2$, the number of ways to orient $\Gamma$ is larger.

(d) If $N^\mathrm{fin}< \infty$ any infinite free paths begin and end with an external edge, provided these are present, which confirms the claim made above. On the other hand, graphs with $N^\mathrm{fin}=\infty$ can be balanced orientable even with an odd number of leads. A trivial example illustrating this claim is a graph $\Gamma$ isometric to the line with the vertices at the points $x=0,1,2,\dots$. A less trivial example is obtained if we replace the positive halfline by a tree graph with all the vertices of degree four. At each branching we have one edge to the left of the vertex and three to the right, of which we choose two of the same orientation as their left neighbor and one opposite. It is easy to see that such a graph with one semi-infinite edge is balanced orientable and it can be identified with a family of infinite free paths.

(e) Admissible orientations can be very different, in particular, if $\sharp I=\infty$. As an example, consider a square-lattice graph corresponding naturally to the subset $ \{(x,y):\: x\in\mathbb{R},\, y=j\ell,\, j\in\mathbb{Z}\} \cup \{(x,y):\: y\in\mathbb{R},\, x=i\ell,\, i\in\mathbb{Z}\}$ of $\mathbb{R}^2$ for a fixed $\ell>0$. One can orient it by making all the horizontal and vertical lines infinite free paths. An alternative is to regard the lattice as a checkerboard pattern and make free paths of the perimeters of all the ``black'' squares; it that case the oriented graph can be identified with the corresponding infinite family of loops. Furthermore, it is not difficult to find ways of orientation in which the two types of free paths, the infinite ones and the loops, are combined.

\end{remarks}
 % ------------- %

\section{Momentum operators}
\label{s:momentum}

After these preliminaries let us pose our main question. We will suppose that $\Gamma$ is the configuration space of a quantum system, namely that a spinless quantum particle lives on the graph, and ask whether one can define for it a momentum-type observable, that is, a self-adjoint operator which acts as $\psi_k \mapsto -i\psi'_k$ on the $k$th edge. The state Hilbert space of such a system will naturally be
 % ------------- %
\begin{equation}
L^2(\Gamma) := \bigoplus_{k\in I^e_\mathrm{fin}} L^2(0,l_k) \,\oplus \bigoplus_{k\in I^e_\infty} L^2(0,\infty)
\end{equation}
 % ------------- %
the elements of which we write as columns $\psi = \big( \{ \psi^\mathrm{fin}_{j,i}:\: j\in I,\, i=1,\dots,d_j^\mathrm{fin} \}, \, \psi^\infty_{j,i}:\: j\in I,\, i=1,\dots,d_j^\infty \} \big)^\mathrm{T}$, or $\psi = \big( \{ \psi^\mathrm{fin}_k:\: k\in I^e_\mathrm{fin} \}\,, \{ \psi^\infty_k:\: k\in I^e_\infty \} \big)^\mathrm{T}$ if we number the edges as we go, with the scalar product
 % ------------- %
\begin{equation}
(\phi, \psi) := \sum_{k\in I^e_\mathrm{fin}} \int_0^{l_k} |\psi^\mathrm{fin}_k(x)|^2 \mathrm{d}x + \sum_{k\in I^e_\infty} \int_0^{\infty} |\psi^\infty_k(x)|^2 \mathrm{d}x\,.
\end{equation}
 % ------------- %
As a starting point of the construction we take the operator $P_0:\, P_0\psi = -i\psi'$ with the domain $D(P_0)$ consisting of $W^{1,2}$ functions vanishing at the vertices, or more explicitly being equal to
 % ------------- %
\begin{equation}
\left\{\psi^\mathrm{fin}_k \in W^{1,2}(e_k): \psi^\mathrm{fin}_k(0)=\psi^\mathrm{fin}_k(l_k)=0 \right\} \cup \left\{\psi^\infty_k \in W^{1,2}(\mathbb{R}^+): \psi^\infty_k(0)=0 \right\}.
\end{equation}
 % ------------- %
It is straightforward to check that $P_0$ defined in this way is symmetric and it adjoint acts as $P^*_0\psi = -i\psi'$ on $D(P^*_0)= W^{1,2}(\Gamma \setminus \mathcal{V})$.

It is natural that candidates for the role of momentum operator are to be looked for among the self-adjoint extensions of $P_0$. To construct them we need the corresponding boundary form; an easy argument using integration by parts gives
 % ------------- %
\begin{eqnarray} \label{boundary-form}
\lefteqn{(P^*\phi,\psi) - (\phi,P^*\psi) = i(\phi^\mathrm{out}, \psi^\mathrm{out}) -i(\phi^\mathrm{in}, \psi^\mathrm{in}) } \nonumber \\ [.7em] && = i \sum_{k\in I^e_\mathrm{fin}} \left( \overline\phi_k^\mathrm{fin}(l_k-) \psi_k^\mathrm{fin}(l_k-) - \overline\phi_k^\mathrm{fin}(0+) \psi_k^\mathrm{fin}(0+) \right) \\ && -i \sum_{k\in {I^e_\infty,+}}  \overline\phi_k^\infty(0+) \psi_k^\infty(0+) +i \sum_{k\in I^e_{\infty,-}} \overline\phi_k^\infty(0-) \psi_k^\infty(0-) \nonumber
\end{eqnarray}
 % ------------- %
for any $\phi,\psi \in D(P^*_0)$, where we have employed the following shorthands,
 % ------------- %
\begin{eqnarray*}
&& \psi^\mathrm{out}:= \left( \psi_1^\mathrm{fin}(0+), \dots, \psi_{N_\mathrm{fin}}^\mathrm{fin}(0+), \psi_1^\infty(0+), \dots, \psi_{N_\infty^+}^\infty(0+) \right)^\mathrm{T}, \\ &&
\psi^\mathrm{in}:= \left( \psi_1^\mathrm{fin}(l_1-), \dots, \psi_{N_\mathrm{fin}}^\mathrm{fin}(l_{N_\mathrm{fin}}-), \psi_1^\infty(0-), \dots, \psi_{N_\infty^-}^\infty(0-) \right)^\mathrm{T}.
\end{eqnarray*}
 % ------------- %

 % ------------- %
 \begin{proposition}
Self-adjoint extensions of $P_0$ always exist provided \mbox{$N_\mathrm{fin}=\infty$;} in case $N_\mathrm{fin}<\infty$ they exist \emph{iff} $N^+_\infty = N^-_\infty$. They are characterized by the condition $\psi^\mathrm{out} = U \psi^\mathrm{in}$ where $U$ is a unitary operator of dimension $N_\mathrm{fin}+N^+_\infty$.
 \end{proposition}
 % ------------- %
 \begin{proof}
It is straightforward to check that the deficiency indices of $P_0$ are $(N_\mathrm{fin}+N^+_\infty, N_\mathrm{fin}+N^-_\infty)$ and that the boundary form \eqref{boundary-form} vanishes \emph{iff} the boundary values satisfy $\psi^\mathrm{out} = U \psi^\mathrm{in}$ for any $\psi\in D(P^*_0)$.
 \end{proof}
 % ------------- %

However, not every self-adjoint extension $P_U$ of $P_0$ is a good candidate for the operators we are interested in. We have to add the requirement of \emph{locality} which means that the boundary conditions couple the boundary values
 % ------------- %
\begin{eqnarray*}
&& \psi_j^\mathrm{out}:= \left( \psi_{j,1}^\mathrm{fin}(0+), \dots, \psi^\mathrm{fin}_{j,d_j^{\mathrm{fin},+}}(0+), \psi_{j,1}^\infty(0+), \dots, \psi^\infty_{j,d_j^{\infty,+}}(0+) \right)^\mathrm{T}, \\ &&
\psi_j^\mathrm{in}:= \left( \psi_{j,1}^\mathrm{fin}(l_{j,1}-), \dots, \psi^\mathrm{fin}_{j,d_j^{\mathrm{fin},-}} (l_{j,d_j^{\mathrm{fin},-}}-), \psi_{j,1}^\infty(0-), \dots, \psi^\infty_{j,d_j^{\infty,-}}(0-) \right)^\mathrm{T},
\end{eqnarray*}
 % ------------- %
where $l_{j,k},\, k=1,\dots,d_j^{\mathrm{fin},-}$ are the lengths of the incoming finite edges at the vertex $v_j$. We will call an extension $P_U$ a \emph{momentum operator} on the graph $\Gamma$ if the unitary operator $U$ is block diagonal, $U= \mathrm{diag\,} \{ U_j:\, j\in I\}$, with the unitary matrix blocks coupling the vertex boundary values
 % ------------- %
\begin{equation} \label{momentum-op-cond}
\psi_j^\mathrm{out} = U_j \psi_j^\mathrm{in}\,, \quad j\in I\,.
\end{equation}
 % ------------- %
This definition leads to the following conclusion.

 % ------------- %
\begin{theorem}
An oriented graph $\Gamma$ supports momentum operators \emph{iff} it is balanced. In such a case they are characterized by the conditions (\ref{momentum-op-cond}).
\end{theorem}
 % ------------- %
\begin{proof}
One has to compare the dimensions of the boundary value spaces.
\end{proof}
 % ------------- %

Before proceeding further it is appropriate to say a few words about the physical meaning of the self-adjointness requirement (\ref{momentum-op-cond}), in particular, in comparison with the more common problem of constructing self-adjoint Laplacians on graphs \cite{ES, Ku, BK}. In the latter case self-adjointness ensures conservation of probability current the components of which, $(\psi_k,-i\psi'_k) = \frac{1}{2i}\, \mathrm{Im}\, (\psi_k,\psi'_k)$ on the $k$-th edge, enter the appropriate boundary form. Here, in contrast, the right-hand side of eq.~(\ref{boundary-form}) contains \emph{probabilities} of finding the particle at the incoming and outgoing edge points; the importance of this fact will become more obvious when we shall construct the unitary groups associated with the operators $P_U$ in Sec.~\ref{s:group} below\footnote{Note also the difference in the size of the matrices which determine the coupling. For a graph Laplacian the extensions are described by $d_j\times d_j$ matrices in the vertex $v_j$, cf. \cite[Thm.~5]{Ku}, while here we deal with unitary matrices which are $d_j^\pm\times d_j^\pm$, i.e. half the size.}.

Another thing which deserves a comment is the locality requirement we have made. We have mentioned in Remark~\ref{rmk:orientability}a that it is sometimes useful to replace a given graph $\Gamma$ by another one in which the vertices are identified; such an identification naturally extends the class of admissible momentum operators. These considerations can have a practical meaning. An elementary example concerns a momentum operator on a finite interval which in itself is not a balanced orientable graph, however, it arises naturally when we factorize it from an infinite periodic system on line \cite[Sec.~III.2]{AGHH}; turning it into a torus --- i.e., a loop in this case --- by identifying the endpoints, we get a family of momentum operators conventionally referred to as \emph{quasimomentum} in this case. A less trivial example of this type is represented by two-interval momentum operators discussed in \cite{JPT11} in connection with the Fuglede conjecture.

Note also that the matrix $U$ defining a momentum operator can be subject to other restrictions, in addition to those imposed by the locality requirement. If there is a proper subgraph $\Gamma' \subset \Gamma$ which can be in the sense of Remark~\ref{rmk:orientability}c identified with a family of paths and $P_U$ is reduced by the subspace $L^2(\Gamma') \subset L^2(\Gamma)$ we say that $P_U$ is \emph{decomposable}, in the opposite case we call the momentum operator \emph{indecomposable}. A decomposable $P_U$ can be regarded as a collection of momentum operators on the appropriate subgraphs which can be analyzed independently.

A similar approach can also be applied locally: if a subspace of the boundary-value space at a given vertex reduces the corresponding coupling matrix $U_j$ we may regard the vertex in question as a family of vertices, each of them connecting only the edges ``talking to each other''. i.e. referring to the same invariant subspace of $U_j$. It is useful to stress that here we mean not just any subspace but such that the vectors associated with the edges involved form its basis, in other words, the matrix $U_j$ is block-diagonal after a suitable permutation of its rows and columns.

\section{Momentum operator spectra: finite graphs}

From now on we shall consider only balanced oriented graphs and investigate properties of momentum operators on them. The first question concerns their spectra; it is not surprising that the graph finiteness plays here a decisive role.

 % ------------- %
\begin{theorem} \label{thm: specfin}
If a graph $\Gamma$ is finite, then any momentum operator $P_U$ on it has a purely discrete spectrum. Moreover, $N_U(\lambda):= \sharp (\sigma(P_U)\cap (-\lambda,\lambda)) \le \frac1\pi L\lambda + \mathcal{O}(1)$ holds as the window half-width $\lambda\to\infty$, where $L$ is the total length of $\Gamma$.
\end{theorem}
 % ------------- %
\begin{proof}
A finite balanced oriented graph can be in view of Remarks~\ref{rmk:orientability} identified with a finite family of loops. Choosing $U_0$ for which all of them are mutually disconnected and the functions on each of them are smoothly connected at the vertices, we get operator $P_{U_0}$ with a purely discrete spectrum --- in fact, one can write explicitly
$\sigma(P_{U_0}) = \bigcup_j \{ \frac{2\pi m}{\ell_j}:\: m\in\mathbb{Z}\}$ where $\ell_j$ are the lengths of the loops. Since by Krein's formula any other self-adjoint extension $P_{U}$ differs from $P_{U_0}$ by a finite-rank perturbation in the resolvent, the character of the spectrum is preserved.

Furthermore, using the explicit spectrum of $P_{U_0}$ we find that the number of its eigenvalues in $(-\lambda,\lambda)$ can be estimated from above by $\sharp I + \sum_j \frac{\lambda \ell_j}{\pi}$ and since by \cite[Sec.~8.3]{We} any other $P_{U}$ can have at most $\sharp I$ eigenvalues, counting multiplicity, in each gap of $\sigma(P_{U_0})$ we arrive at the estimate
 % ------------- %
 $$
N_U(\lambda) \le 3\sharp I + \frac1\pi \lambda \sum_j \ell_j\,,
 $$
 % ------------- %
which yields the sought `semiclassical' result.
\end{proof}
 % ------------- %

 % ------------- %
\begin{example}
Consider a \emph{loop graph} of total length $L$ with vertices at points $0<x_1<x_2< \cdots <x_N=L$ with the points $x=0$ and $x=L$ identified. Since the degree of each vertex is two, all the admissible momentum operators are characterized by $N$-tuples $\alpha=\{\alpha_1,\dots,\alpha_N\} \in (-\pi,\pi]^N \subset \mathbb{R}^N$ through the conditions $\psi(x_j+) = \mathrm{e}^{i\alpha_j} \psi(x_j-)\,, \: j=1,\dots,N$. The spectrum is easily found to be
 % ------------- %
 $$
\sigma(P_U) = \bigg\{ \frac1L \Big( 2\pi n - \sum_{j=1}^N \alpha_j \Big):\: n\in \mathbb{Z}\, \bigg\}\,.
 $$
 % ------------- %
It is obvious that the family of all $P_U$ decomposes into isospectral classes specified by the values $\det U$ of the diagonal matrices $U$.
\end{example}
 % ------------- %
The spectrum may become more complex once the graph $\Gamma$ has a branching, i.e. there is at least one vertex of degree $d_j\ge 4$; recall that $d_j=3$ is excluded by the orientability requirement as discussed in Sec.~\ref{s: balanced}.

 % ------------- %
\begin{example}
A \emph{figure-eight graph} is equivalent in the sense of \cite{Ku, EL} to the ``two-interval'' case discussed in detail in \cite{JPT11}. The eigenvalues of any $P_U$ here have maximum multiplicity two but their distribution is more complicated than in the previous example, in particular, it is purely aperiodic provided the length ratio of the two loops of the graph is irrational.
\end{example}
 % ------------- %

The asymptotic estimate given in Theorem~\ref{thm: specfin} can be in fact made two-sided.
 % ------------- %
\begin{theorem} \label{thm: specfin2}
For a finite $\Gamma$ we have $N_U(\lambda) = \frac1\pi L\lambda + \mathcal{O}(1)$ as $\lambda\to\infty$.
\end{theorem}
 % ------------- %
\begin{proof}
We can identify the loops from the proof of Theorem~\ref{thm: specfin} with the edges of $\Gamma$ in the sense of Remark~\ref{rmk:orientability}a. Substituting the boundary values $\psi^\mathrm{out} = \{ c_1,\dots, c_{N_\mathrm{fin}} \}^\mathrm{T}$ and $\psi^\mathrm{in} = \{ c_1 \mathrm{e}^{ik\ell_1},\dots, c_{N_\mathrm{fin}} \mathrm{e}^{ik\ell_{N_\mathrm{fin}}} \}^\mathrm{T}$ into the condition (\ref{momentum-op-cond}) we get $N_\mathrm{fin}$ conditions to satisfy, $\sum_{m=1}^{N_\mathrm{fin}} \left( U_{jm} \mathrm{e}^{ik\ell_m} - \delta_{jm} \right) c_m = 0\,,\, j=1,\dots,N_\mathrm{fin}\,$; note that the left-hand side can vanish at real values of $k$ only since $U$ is unitary. This system of linear equations can be solved \emph{iff}
 % ------------- %
 $$
\det\left( U_{jm} \mathrm{e}^{ik\ell_m} - \delta_{jm} \right)_{j,m=1}^{N_\mathrm{fin}} =0\,.
 $$
 % ------------- %
The left-hand side of the last relation is a trigonometric polynomial with the senior and junior terms $\prod_j U_{jj}\, \mathrm{e}^{ikL}$ and $(-1)^{N_\mathrm{fin}}$, respectively. Assume first that the senior coefficient $\prod_j U_{jj}\ne 0$. According to Langer's classical result \cite{La} the number of zeros in the interval $(-\lambda,\lambda)$ behaves then as $\frac{1}{\pi} L\lambda + \mathcal{O}(1)$ as $\lambda\to \infty$. If the assumption is not satisfied we use the fact that by the implicit-function theorem the zeros as functions of the entries $U_{jm}$ of $U$ are analytic, hence for such a $P_U$ there is a family of zeros with the `correct' asymptotics arbitrarily close to $\sigma(P_U)$.

Let us add that there is a subtle point in the last argument. In general it may happen that the vanishing senior coefficient in a spectral condition in the form of a trigonometric expression may change the asymptotic distribution of zeros if some of the latter escape to infinity in the complex plane as the parameters vary; an example of such a behavior can be found in \cite{DEL}. Here it cannot happen, however, because we know from the outset that all the zeros lie at the real axis.
\end{proof}
 % ------------- %

 \section{Momentum operator spectra: infinite graphs}

The situation changes substantially if we abandon the finiteness assumption.
 % ------------- %
\begin{theorem}
If $N_\infty>0$ we have $\sigma(P_U)= \sigma_\mathrm{ess}(P_U) =\mathbb{R}$ for any momentum operator $P_U$. The same is true if $\sharp I=\infty$ and there is a unitary operator $U'$ with $\dim \emph{Ran} (U-U')<\infty$ such that $P_{U'}$ is decomposable containing at least one infinite free path; in such a case we also have  $\sigma_\mathrm{ac}(P_U) =\mathbb{R}$.
\end{theorem}
 % ------------- %
\begin{proof}
Suppose that the graph has at least one external edge. Since the graph remains balanced oriented if we switch the orientation of each edge, we may assume without loss of generality that it is an outgoing one parametrized by $[0,\infty)$. For a fixed $k\in\mathbb{R}$ we take the following family of functions,
 % ------------- %
\begin{equation}
\psi_{y,\epsilon}:\: \psi_{y,\epsilon}(x) = \epsilon^{1/2}\, \mathrm{e}^{ikx}\, \phi(\epsilon(x-y))\,, \quad \epsilon>0\,,
\end{equation}
 % ------------- %
with $\phi \in C_0^\infty(-1,1)$ such that $\|\phi\|=1$. We obviously have $\|\psi_{y,\epsilon}\|=1$ and the supports lie in the positive halfline provided $y>\epsilon^{-1}$. We use the same symbol for functions on $\Gamma$ which are zero on all the other edges; them $P_U \psi_{y,\epsilon}$ is independent of $U$ and an easy computation gives
 % ------------- %
 $$ %\begin{equation}
\| (P_U-k) \psi_{2\epsilon^{-1},\epsilon} \|^2 = \epsilon^2 \int_{-1}^1 |\phi'(u)|^2\, \mathrm{d}u \, \longrightarrow\, 0
 $$ %\end{equation}
 % ------------- %
as $\epsilon\to 0$. Choosing $\epsilon_n = 2^{-2n},\, n\in\mathbb{N}$, we achieve that the functions with different $n$ have disjoint supports, hence $\psi_{2\epsilon_n^{-1},\epsilon_n} \to 0$ in the weak topology of $L^2(\Gamma)$, and by Weyl's criterion we infer that $k\in \sigma_\mathrm{ess} (P_U)$.

In the second case the motion on $\Gamma$ corresponding to $P_{U'}$ can be decoupled and it contains a component which is unitarily equivalent to the momentum operator on $\mathbb{R}$, hence $\sigma_\mathrm{ac}(P_{U'}) =\mathbb{R}$; since by assumption the resolvents of $P_{U'}$ and the original $P_U$ differ by a finite-rank operator, their absolutely continuous spectra coincide.
\end{proof}
 % ------------- %

The second claim can ensure that the absolutely continuous spectrum covers the real axis even in absence of external edges; examples are easily found. On the other hand, conclusions of the theorem naturally do not mean  that momentum operators on finite-core graphs must have a purely absolutely continuous spectrum.

 % ------------- %
\begin{example} \label{ex:loopline}
Consider a graph consisting of a line to which a loop of length $\ell>0$ is attached at one point. Modulo mirror transformations, there is essentially one way to orient such a graph in a balanced way. Consider two different matrices coupling the boundary values $(\psi_1^\infty(0-), \psi^\mathrm{fin}(\ell-))^\mathrm{T}$ and $(\psi_2^\infty(0+), \psi^\mathrm{fin}(0+))^\mathrm{T}$. If we take $U= {0\;1 \choose 1\;0}$ the corresponding $P_U$ is unitarily equivalent to the momentum operator on the line and has a purely \emph{ac} spectrum. On the other hand, $U'=I$ leads to full decoupling of the loop from the line giving rise to an infinite family of  eigenvalues, $\{ \frac{2\pi m}{\ell}:\: m\in\mathbb{Z}\}$ embedded in $\sigma_\mathrm{ac}(P_{U'})=\mathbb{R}$.
\end{example}
 % ------------- %

Different ways of orienting a given undirected graph can give rise to momentum operators with different spectra. Nevertheless, some properties are invariant with respect to the choice of orientation, for instance, the conclusions related to the the presence or absence of external edges in the above two theorems. The differences can be more dramatic if $\sharp I=\infty$ as the following example shows.

 % ------------- %
\begin{example}
Consider the square-lattice graph of Remark~\ref{rmk:orientability}e. If the orientation follows the horizontal and vertical lines we can choose $U$ which identifies at each vertex the limits in the two directions separately. The respective operator $P_U$ is isomorphic to an infinite direct sum of identical copies of the operator $-i \frac{\mathrm{d}}{\mathrm{d}x}$ on $L^2(\mathbb{R})$, and consequently, $\sigma(P_U)= \mathbb{R}$ with infinite multiplicity. On the other hand, consider the ``checkerboard'' orientation which identifies the graph with an infinite family of loops of length $4\ell$. Choosing $U'$ which gives rise to self-adjoint extension $P_{U'}$ the domain of which are functions which are $W^{1,2}$ locally and continuous on each square loop, we get $\sigma(P_{U'})= \{ \frac{\pi m}{2\ell}:\: m\in\mathbb{Z}\}$, again infinitely degenerate.
\end{example}
 % ------------- %

 \section{Groups associated with momentum operators}
 \label{s:group}

In order to describe the associated groups let us first introduce the notion of a \emph{route} $r_{x',x}$ from a point $x'\in \Gamma$ to a point $x\in \Gamma$. By this we mean a finite curve $r_{x',x}:\: [0,\ell(r_{x',x})] \to\Gamma$ with $r_{x',x}(0) = x$ and $r_{x',x}(\ell(r_{x',x}) = x'$ respecting the orientation of $\Gamma$; the number $\ell(r_{x',x})$ is called the \emph{length} of $r_{x',x}$. We stress the difference between a \emph{path} in $\Gamma$ considered in Section~\ref{s: balanced}  above and a route: the former is a sequence of edges through which the particle can travel respecting the orientation, the latter is the actual journey it made. If $\Gamma$ contains a loop, for instance, for a path it is not important whether one circles it repeatedly or not, while for a route it certainly makes a difference.

We fix a momentum operator $P_U$ on $\Gamma$ determined by a matrix $U$, and define the \emph{route factor} for $r_{x',x}$ with $x',x\not\in\mathcal{E}$. It is equal to one if there is no $x'' \in (0,\ell(r_{x',x}))$ such that $r_{x',x}(x'')\in \mathcal{E}$, i.e. if a traveler on this route meets no vertex. In the opposite case we take in each vertex on the route the factor coming from the coupling --- if the traveler passes from $m$-th ingoing edge to the $j$-th outgoing the factor is the $U_{jm}$ element of the matrix --- and define the route factor as the product of the factors coming from all the vertices passed on the route; we denote it as $U(r_{x',x})$. It may be zero if some of the matrix elements on the way vanishes.

Given $x\in\Gamma$ which is not a vertex we denote by $\mathcal{R}_\mathit{a}(x)$ the set of routes of length $\mathit{a}$ ending at $x$. It may consists of a single route which happens if going back we meet no vertex --- it happens for sure if $\mathit{a}$ is small enough --- or of a larger number of routes, however, in view of the assumption (\ref{minlength}) their number is finite. Equipped with these notions we can construct the group associated with $P_U$.

We will write elements of $L^2(\Gamma)$ as $\psi = \{ \psi_k:\: k\in I^e\}^\mathrm{T}$ without distinguishing now internal and external edges. Choosing a point $x =\{x_k:\: x_k\in e_k,\, k\in I^e\}$ we write values of a function representing element $\psi\in L^2(\Gamma)$ at the point $x$ as $\psi(x) = \{ \psi_k(x_k):\: k\in I^e\}^\mathrm{T}$ and define the operator $\mathcal{U}(\mathit{a}):\: L^2(\Gamma) \to L^2(\Gamma)$ by
 % ------------- %
\begin{equation} \label{group}
\big(\mathcal{U}(\mathit{a})\psi \big)(x) := \bigg\{ \sum_{r_{x'_{k'},x_k} \in \mathcal{R}_\mathit{a}(x_k)} U(r_{x'_{k'},x_k}) \psi_{k'}(x'_{k'}):\: k\in I^e \bigg\}^\mathrm{T}
\end{equation}
 % ------------- %
for those $x\in\Gamma$ and $\mathit{a}$ for which none of the points $x_k$ and $x'_k$ involved coincides with a vertex of $\Gamma$; note that for each $\mathit{a}\in\mathbb{R}$ this requirement excludes an at most countable subset of graph points. It is easy to check that these operators form a group, $\mathcal{U}(\mathit{a}) \mathcal{U}(\mathit{a})'= \mathcal{U}(\mathit{a}+\mathit{a}')$ for all $\mathit{a}, \mathit{a}'\in \mathbb{R}$. It is less obvious that the operator defined by (\ref{group}) is unitary. To check it let us first look how does the map
 % ------------- %
 $$
\psi_k \mapsto \sum_{r_{x'_{k'},x_k} \in \mathcal{R}_\mathit{a}(x_k)} U(r_{x'_{k'},x_k}) \psi_{k'}
 $$
 % ------------- %
act in the space $\ell^2(I^e)$. The set $\mathcal{R}_\mathit{a}(x)$ can be regarded as a `backward' tree with every route from `tip' $x'_{k'}$ to the `root' $x_k$ having length $\mathit{a}$. Since by construction the incoming and outgoing edges at every vertex are related by a unitary matrix, the norm is preserved and images of orthogonal elements remain orthogonal going `backward' through the tree. Consequently, we have
 % ------------- %
 $$
\bigg| \sum_{r_{x'_{k'},x_k} \in \mathcal{R}_\mathit{a}(x_k)} U(r_{x'_{k'},x_k}) \psi_{k'}(x'_{k'}) \bigg|^2 = |\psi_k(x_k)|^2
 $$
 % ------------- %
and contributions to $\mathcal{U}(\mathit{a})\psi$ coming from different edge components are mutually orthogonal. Since the integration over $x_k$ and $x'_{k'}$ on each edge is taken with respect to the same Lebesgue measure, it follows that $\|\mathcal{U}(\mathit{a})\psi\|^2 = \|\psi\|^2$ holds for any $\psi\in L^2(\Gamma)$ which we wanted to prove.

\smallskip

One has to keep in mind that the group depends on the choice of the operator $P_U$ and if there is need to stress this fact we should write its elements as $\mathcal{U}_U(\mathit{a})$. In general the group actions are shifts along the graph edges, however, in contrast to the trivial situation when $\Gamma$ is a line, vertex coupling plays a role, especially if there is a nontrivial branching, $d_j>2$. This returns us to the question of physical meaning of different self-adjoint extension we have touched briefly in Sec.~\ref{s:momentum}.

 % ------------- %
\begin{example}
To illustrate the action the group consider a \emph{star-shaped graph} consisting of an equal number $n$ of incoming and outgoing semi-infinite leads, $N^\pm_\infty=n\ge 2$, connected in a single vertex, in which the coupling is described by an $n\times n$ unitary matrix $U$. Take a vector $\psi$ on the first incoming lead identified with the interval $(-\infty,0]$ assuming that it has a compact support and denote $b:= - \inf\,\mathrm{supp}\,\psi$. If $\mathit{a}>b$ the function $\mathcal{U}(\mathit{a})\psi$ is supported on the outgoing leads only, its component on the $j$-th one being $U_{j1}\psi(\cdot-\mathit{a})$. In other words, the incoming wave packets splits into scaled copies of the original one with the weights which guarantee that the probability after the shift through the vertex\footnote{We avoid the term ``passing through'' to stress that no time evolution is involved here, at least as long as we think about our model in terms of standard quantum mechanics.} is preserved. If the initial wave packet is supported on more than a single edge, the resulting one is naturally a superposition of those coming from the involved incoming contributions.
\end{example}
 % ------------- %

More generally, it is easy to see that a state represented by a function of compact support will remain compactly supported if operators $\mathcal{U}(\mathit{a})$ are applied to it. This does not mean, however, that the support will keep its properties, in particular, that its Lebesgue measure should have a bound independent of $\mathit{a}$.

 % ------------- %
\begin{example}
Consider again the graph of Example~\ref{ex:loopline}. In the two situation mentioned there the group action is simple: for $U= {0\;1 \choose 1\;0}$ the operator $P_U$ is unitarily equivalent to the momentum operator on the line along which the corresponding $\mathcal{U} (\mathit{a})$ shifts function, while for $U=I$ the group decomposes into shifts along the line and cyclic motion on the loop. The situation is different if all the elements of $U$ are nonzero. Consider $\psi$ with the support on the incoming lead, $\mathrm{supp}\,\psi \subset \big( -\frac34\ell, -\frac14\ell \big)$ and apply $\mathcal{U}(n\ell)$ to it. It is not difficult to see that the loop component keeps the shape changing just the ``size'' being $\big(\mathcal{U}(\mathit{n\ell}) \psi\big)_2(x) = (U_{22})^{n-1} U_{21}\psi(x-\ell)$ while
 % ------------- %
 $$
\big(\mathcal{U}(\mathit{n\ell}) \psi\big)_1(x) = U_{11}\psi(x-n\ell) + \sum_{k=1}^{n-1} U_{12}(U_{22})^{n-1-k} U_{21}\psi(x-k\ell)
 $$
 % ------------- %
holds for the line part. In other words, the function is a linear combination of shifted copies of the original function with the number of components increasing with $n$, in particular, we have $\:\mathrm{diam}\, \mathrm{supp}\, \big(\mathcal{U}(\mathit{n\ell}) \psi\big)_1 \ge (n-1)\ell$.
\end{example}
 % ------------- %

The action of $\mathcal{U}(\mathit{a})$ on an indecomposable graph, even a finite one, can be more more complicated if such a $\Gamma$ has more than one loop and their lengths are incommensurable. Note also that a related example can be found in \cite{JPT12a} where, however, the model is interpreted in terms of Lax-Phillips scattering theory.

\section{A two-loop example}
\label{s:example}

To get another insight into spectral properties of the operators $P_U$, let us analyze one more simple example. The graph $\Gamma$ in this case will consists of two leads, $\Gamma_0$ and $\Gamma_4$ identified with the halflines $(-\infty,0]$ and $(]0,\infty)$, respectively, and three finite edges $\Gamma_j$ of lengths $\ell_j,\: j=1,2,3,$ connecting the endpoints of $\Gamma_0$ and $\Gamma_4$; the first two are oriented from $\Gamma_0$ to $\Gamma_4$ and the third one in the opposite direction. Such a graph is obviously balanced oriented so one can construct momentum operators on it; we choose the one for which the coupling at both vertices (coupling respectively the edges $(3,0)\to(1,2)$ and $(1,2)\to(3,4)$) is described by the same matrix,
 % ------------- %
\begin{equation}
U= U^*= \frac{1}{\sqrt{2}} {1\quad\; 1 \choose 1\; -\!1}\,.
\end{equation}
 % ------------- %
Looking for eigenfunctions, proper or generalized, of the momentum operator --- which we denote for the sake of simplicity $P_U$ again --- we use the Ans\"atze $\psi_j(x) = c_j\, \mathrm{e}^{-ikx}$ on the $k$-th edge, $j=0,\dots,4$.

Consider first the situation when the two ``right-looking'' edges are of the same length, $\ell_1=\ell_2$. Specifying the conditions \eqref{momentum-op-cond} to the present case and excluding the coefficients $c_1,\,c_2$ we arrive at the relation
 % ------------- %
 $$
\mathrm{e}^{-ik\ell_1} {c_3\choose c_4} = {c_3\,\mathrm{e}^{ik\ell_3} \choose c_0}\,.
 $$
 % ------------- %
It has two possible independent solutions. If $c_0=c_4$ we require $c_3= c_3\,\mathrm{e}^{ik\ell_3}$ which gives rise to an infinite series of eigenvalues,
 % ------------- %
\begin{equation} \label{ex:evs}
k = \frac{2\pi n}{\ell_1+\ell_3}\,, \quad n\in\mathbb{Z}\,;
\end{equation}
 % ------------- %
the coefficients for the remaining edges are $c_1=c_2 = \frac{1}{\sqrt{2}}\, c_3\, \mathrm{e}^{-ik\ell_1}$. On the other hand, if $c_3=0$ the solution exists for any $k\in \mathbb{R}$ with the coefficients $c_4= c_0\, \mathrm{e}^{ik\ell_1}$ and $c_1=-c_2 = \frac{1}{\sqrt{2}}\, c_0$. In other words, the spectrum of $P_U$ consists of absolutely continuous part covering the real axis and the series of embedded eigenvalues \eqref{ex:evs}.

This looks like the spectrum we have found in Example~\ref{ex:loopline}, however, there is an important difference. The operator $P_{U'}$ there was decomposable, hence its embedded point spectrum was in a sense trivial. Here it is not the case and still we have compactly supported eigenfunctions on a infinite graph. Moreover, the example can be modified by replacing the two semi-infinite leads by a single finite edge of length $\ell_4$ connecting the two vertices again. The solutions symmetric with respect to permutation of $\psi_j,\: j=1,2$, are as before, with eigenfunctions vanishing on the new edge. The antisymmetric ones, on the other hand, require the relations $c_0= c_4\, \mathrm{e}^{ik\ell_4} = c_0\, \mathrm{e}^{ik(\ell_1+\ell_4)}$ to be satisfied, giving rise to eigenvalues
 % ------------- %
$$
k = \frac{2\pi n}{\ell_1+\ell_4}\,, \quad n\in\mathbb{Z}\,,
$$
 % ------------- %
replacing the absolutely continuous spectral component of the original example. What is important is that they correspond to $c_3=0$, hence all the eigenfunctions in this case vanish at some edge of the modified, now compact graph. These examples allow us to make the following conclusion.

 % ------------- %
 \begin{proposition}
For momentum operators on a balanced graph the unique continuation principle is in general not valid.
 \end{proposition}
 % ------------- %

We note that a similar result holds for graph Laplacians. Here, however, the claim is less obvious; Laplacian eigenfunctions corresponding to positive eigenvalues are trigonometric functions and as such they can have zeros which for a proper choice of geometry may coincide with graph vertices \cite{EL}, while momentum operator eigenfunctions on an edge cannot vanish being of the form $c\, \mathrm{e}^{ikx}$.

Returning to the original example we note that the embedded spectrum emerged as the result of the assumption $\ell_1=\ell_2$. If it is not valid, the problem is more complicated. Writing down the coupling conditions and excluding $c_1,\,c_2$ we arrive at
 % ------------- %
$$
c_3 = c_0\, \frac{\mathrm{e}^{ik\ell_2} - \mathrm{e}^{ik\ell_1}}{\mathrm{e}^{ik\ell_3} \big(\mathrm{e}^{ik\ell_2} + \mathrm{e}^{ik\ell_1} \big) -2}
$$
 % ------------- %
together with the condition $c_0\, F\big(\mathrm{e}^{ik\ell_1}, \mathrm{e}^{ik\ell_2}, \mathrm{e}^{ik\ell_3}\big) = c_4$ where $F$ is a rational function with same denominator. Embedded eigenvalues correspond to the $k$'s for which the latter vanishes, hence such a spectrum is present under suitable commensurability relations between the edges lengths.

One can naturally also ask what happens if we change those parameters. In case of graph Laplacians we know that violation of rationality turns in general embedded eigenvalues into resonances. Let us look what we have here; for simplicity we consider the situation when $\ell_1=\ell_3=\ell$ and $\ell_2 = \ell+\delta$. The mentioned denominator vanishes when $4i\sin k\ell + \mathrm{e}^{ik\ell} \big( \mathrm{e}^{ik\delta} -1\big)$ does, hence for small $\delta$ it yields
 % ------------- %
\begin{equation} \label{ex:res}
k = \frac{\pi n}{\ell} \left( 1+ \frac{(-1)^n}{4i}\, \delta + \mathcal{O}(\delta^2) \right)
\end{equation}
 % ------------- %
in the vicinity of the eigenvalues \eqref{ex:evs}, i.e. complex ``resonance'' solutions; in distinction to Laplacian resonances one cannot localize them in one complex halfplane.

\bibliographystyle{amsalpha}

\begin{thebibliography}{A}

\bibitem [AGHH]{AGHH}
S.~Albeverio, F.~Gesztesy, R.~Hoegh-Krohn, and H.~Holden,
\textit{Solvable Models in Quantum Mechanics, 2nd ed. with appendix by P. Exner}, AMS Chelsea, Rhode Island, 2005.

\bibitem [BK]{BK}
G.~Berkolaiko and P.~Kuchment,
\textit{An Introduction to Quantum Graphs}, a book in preparation.

\bibitem [Ca99]{Ca}
R.~Carlson: Inverse eigenvalue problems on directed graphs, \emph{Trans. Am. Math. Soc.} \textbf{351} (1999), 4069--4088.

\bibitem [DEL10]{DEL}
E.B.~Davies, P.~Exner, J.~Lipovsk\'{y}, \textit{Non-Weyl asymptotics for quantum graphs with general coupling conditions}, J. Phys. A: Math. Theor. {\bf 43} (2010), 474013.

\bibitem [AEG]{AEG}
G.~Dell'Antonio, P.~Exner, V.~Geyler, eds., \textit{Special Issue on Singular Interactions in Quantum Mechanics: Solvable Models}, J. Phys. A: Math. Gen. {\bf 38} (2005), No.~22.

\bibitem [Eg11]{Eg}
S.~Egger \emph{n\'e} Endres: The solution of the ``constant term problem'' and the $\zeta$-regularized determinant for quantum graphs, \emph{PhD thesis}, Universit\"at Ulm 2011.

\bibitem [ES10]{ES10}
S.~Endres~, F.~Steiner: The Berry–Keating operator on $L^2(\mathbb{R}>, \mathrm{d}x)$ and on
compact quantum graphs with general self-adjoint realizations, J. Phys. A: Math. Theor. \textbf{43} (2010), 095204.

\bibitem[EKKST]{EKKST}
P.~Exner, J.P.~Keating, P.~Kuchment, T.~Sunada, A.~Teplayaev, eds., \emph{Analysis on graphs and its applications}, Proc. Symp. Pure Math., vol. 77; Amer. Math. Soc., Providence, R.I., 2008.

\bibitem [EL10]{EL}
P.~Exner, J.~Lipovsk\'{y}, \textit{Resonances from perturbations of quantum graphs with rationally related edges}
J. Phys. A: Math. Theor. \textbf{43} (2010), 105301.

\bibitem [E\v{S}89]{ES}
P.~Exner, P.~\v{S}eba, \textit{Free quantum motion on a branching graph}, Rep. Math. Phys. {\bf 28} (1989), 7--26.

\bibitem [FKW07]{FKW}
S.A.~Fulling, P.~Kuchment, J.H.~Wilson, \textit{Index theorem for quantum graphs}, J. Phys. A: Math. Theor. \textbf{40} (2007), 14165--14180.

\bibitem [JPT11]{JPT11}
P.E.T.~Jorgensen, S.~Pedersen, Feng Tian, \textit{Momentum operators in two intervals: spectra and phase transition},
\texttt{arXiv:1110.5948}.

\bibitem [JPT12a]{JPT12a}
P.E.T.~Jorgensen, S.~Pedersen, Feng Tian, \textit{Translation representations and scattering by two intervals},
J. Math. Phys. \textbf{53} (2012), 053505

\bibitem [JPT12b]{JPT12b}
P.E.T.~Jorgensen, S.~Pedersen, Feng Tian, \textit{Spectral theory of multiple intervals},
\texttt{arXiv:1201.4120}

\bibitem [KS03]{KS}
T.~Kottos, U.~Smilansky, \textit{Quantum graphs: a simple model for chaotic scattering}, J. Phys. A: Math. Gen. {\bf 36} (2003), 3501--3524.

\bibitem [Ku08]{Ku}
P. Kuchment, \textit{Quantum graphs: an introduction
and a brief survey}, in the proceedings volume \cite{EKKST}, pp.~291--314.

\bibitem[La31]{La}
R.E.~Langer, \textit{On the zeros of exponential sums and
integrals}, Bull.  Amer. Math. Soc. \textbf{37} (1931), 213--239.

\bibitem[We]{We}
J.~Weidmann, \textit{Linear Operators in Hilbert Space}, Springer, New York 1980.





\end{thebibliography}

\end{document}